\documentclass[12pt]{article}
\usepackage{fullpage,epsfig,graphics,amsbsy,amssymb,cancel,bbm}
\usepackage{hyperref,color,slashed}
\usepackage{graphicx}
\usepackage{amsfonts,mathdots,dsfont}
\usepackage{amssymb,amsmath,amscd,mathrsfs}
\usepackage{tikz}
\usetikzlibrary{arrows,chains,matrix,positioning,scopes,snakes}

\makeatletter
\tikzset{join/.code=\tikzset{after node path={%
\ifx\tikzchainprevious\pgfutil@empty\else(\tikzchainprevious)%
edge[every join]#1(\tikzchaincurrent)\fi}}}
\makeatother

\tikzset{>=stealth',every on chain/.append style={join},
         every join/.style={->}}
\tikzset{
    >=stealth',
    punkt/.style={
           rectangle,
           rounded corners,
           draw=black, very thick,
           text width=6.5em,
           minimum height=2em,
           text centered},
    pil/.style={
           ->,
           thick,
           shorten <=2pt,
           shorten >=2pt,}
}

\newcommand{\BB}{\mathbb}

\def\L{{\cal L}}

\newcommand{\bea}{\begin{eqnarray}}
\newcommand{\eea}{\end{eqnarray}}
\newcommand{\be}{\begin{eqnarray}}
\newcommand{\ee}{\end{eqnarray}}
\newcommand{\nn}{\nonumber}
\newcommand{\Tr}{\textrm{Tr}}

\newcommand{\sbullet}{\textrm{\tiny{\textbullet}}}
\newcommand{\bra}{\langle}
\newcommand{\ket}{\rangle}

\newcommand{\im}{\textrm{Im}\,}
\newcommand{\re}{\textrm{Re}\,}
\newcommand{\To}{\Rightarrow}
\newcommand{\reeb}{\textrm{\scriptsize{$R$}}}
\newcommand{\sreeb}{\textrm{\tiny{$R$}}}

\def\ga{\alpha}

\def\gc{\gamma}
\def\Gc{\Gamma}
\def\Gd{\Delta}
\def\gd{\delta}
\def\ep{\epsilon}

\def\gs{\sigma}
\def\Gs{\Sigma}
\def\gk{\kappa}
\def\gl{\lambda}

\def\Go{\Omega}
\def\go{\omega}

\DeclareMathAlphabet{\mathpzc}{OT1}{pzc}{m}{it}

\newtheorem{theorem}{Theorem}[section]
\newtheorem{lemma}[theorem]{Lemma}

\newenvironment{proof}[1][Proof]{\begin{trivlist}
\item[\hskip \labelsep {\bfseries #1}]}{\end{trivlist}}

\newenvironment{remark}[1][Remark]{\begin{trivlist}
\item[\hskip \labelsep {\bfseries #1}]}{\end{trivlist}}

\newcommand{\qed}{\nobreak \ifvmode \relax \else
      \ifdim\lastskip<1.5em \hskip-\lastskip
      \hskip1.5em plus0em minus0.5em \fi \nobreak
      \vrule height0.5em width0.5em depth0.00em\fi}

\begin{document}
\thispagestyle{empty}
\begin{flushright} \small
UUITP-09/14
 \end{flushright}
\smallskip
\begin{center} \LARGE
{\bf On twisted $N=2$ 5D super Yang-Mills theory}
 \\[12mm] \normalsize
{\bf  Jian Qiu$^{a,b}$ and Maxim Zabzine$^c$} \\[8mm]
 {\small\it
${}^a$Math\'ematiques, Universit\'e du Luxembourg,\\
 Campus Kirchberg, G 106,  L-1359 Luxembourg\\
      \vspace{.3cm}
${}^b$Max-Planck-Institut f\"ur Mathematik,\\
 Vivatsgasse 7,
53111 Bonn, Germany\\
      \vspace{.5cm}
${}^c$Department of Physics and Astronomy,
     Uppsala university,\\
     Box 516,
     SE-75120 Uppsala,
     Sweden\\
   }
\end{center}
\vspace{7mm}
\begin{abstract}
 \noindent
 On a five dimensional simply connected Sasaki-Einstein manifold, one can construct Yang-Mills theories coupled to matter with at least two supersymmetries.  The partition function of these theories localises on the contact instantons, however the contact instanton equations are not elliptic.  It turns out that these equations can be embedded into the Haydys-Witten equations (which are elliptic) in the same way the 4D anti-self-dual instanton equations are embedded in the Vafa-Witten equations. We show that under some favourable circumstances, the latter equations will reduce to the former by proving some vanishing theorems. It was also known that the Haydys-Witten equations on  product manifolds $M_5=M_4\times \BB{R}$  arise in the context of twisting the 5D maximally supersymmetric Yang-Mills theory. In this paper, we present the construction of twisted $N=2$ Yang-Mills theory on Sasaki-Einstein manifolds, and more generally on $K$-contact manifolds.
 The localisation locus of this new theory thus provides a covariant version of the Haydys-Witten equation.
\end{abstract}

\eject
\normalsize
\tableofcontents
\section{Introduction}\label{sec_intro}
Initiated by the Pestun's work \cite{Pestun:2007rz} the localization technique has been widely used for the exact
calculation of partition functions and other supersymmetric observables for the supersymmetric gauge theories
 in different dimensions.  In this work we are interested in analysing further the equations governing the localisation locus
  in 5D gauge theories, especially their implication on contact geometry.

 In the work \cite{Kallen:2012cs} (based on earlier works \cite{Nekrasov:1996cz} and  \cite{Baulieu:1997nj})
  a twisted version of 5D Yang-Mills theory has been constructed and the theory  localizes on contact instantons
\bea
 \star F = - \kappa \wedge F~,\label{contact-inst}
\eea
 where $\kappa$ is the contact form on the five dimensional manifold. These equations appeared previously also in \cite{Harland:2011zs}
  (see also  \cite{Corrigan:1982th} for first discussion of self-duality in higher dimensions and \cite{Fan} for the discussion of the odd dimensional case).
 On any simply connected Sasaki-Einstein manifolds, one can construct the $N=1$ 5D supersymmetric Yang-Mills theory with matter that possesses at least
 two supersymmetries.
 On special classes of these manifolds, namely the toric ones, the perturbative (zero instanton) part of the partition
 function can be calculated quite efficiently, see \cite{Kallen:2012va} for $S^5$,  \cite{Qiu:2013pta, Qiu:2013aga} for $Y^{p,q}$
 and \cite{Qiu:2014oqa} for more general instances of such manifolds.
 The full partition function for $S^5$  has been conjectured in  \cite{Kim:2012ava, Lockhart:2012vp, Kim:2012qf}
 and for the case of general simply connected toric Sasaki-Einstein manifolds in \cite{Qiu:2014oqa}. A
 derivation of the partition function based on first principle is not available so far, and it will require a better understanding of the contact instanton equations (\ref{contact-inst}).

   In fact, the twisted $N=1$ supersymmetric complex can be extended to any $K$-contact manifolds, and the equations (\ref{contact-inst}) arising out of the localisation have been investigated subsequently in \cite{Wolf:2012gz,Baraglia:2014gma,Pan:2014nha}, due to their potential bearing on the contact geometry. Also, some interesting observations have been made about the property of the perturbative part of the partition function, namely one can construct the partition function on the whole manifold by gluing together copies of the partition function on $\BB{R}^4\times S^1$, one copy for each closed Reeb orbit \cite{Qiu:2013pta, Qiu:2013aga} (the flow of the Reeb vector field does not in general form closed orbits, but usually there are isolated closed Reeb orbits for the generic choice of $\kappa$). This observation leads inevitably to the conjecture that the instanton sector enjoys the same property and can be constructed by gluing together flat space results. But the whole conjecture is predicated on the assumption that when one has only isolated Reeb orbits, the instantons tend to concentrate along those orbits, behaving as if they were point like particles propagating along the closed Reeb orbits. There is yet no satisfactory argument for this (see \cite{Pan:2014bwa} for some related observation), but the conjecture is echoed by what one has observed in 4D, where Pestun showed that there are no smooth instantons on $S^4$ but one has to add the point like solutions at the two poles \cite{Pestun:2007rz}. Also in 3D, where Taubes proved that the solution to the vortex equation is concentrated along the closed Reeb orbits, see \cite{Taubes_SWW}.

   Despite the supporting evidence, the lack of control of the analytical behaviour of the instantons has left the calculation for the instanton sector of the 5D case incomplete, compared to the more conclusive result obtained on $S^4$. As a first step to address this problem, we will embed the contact instanton equations into a larger set of equations known as the Haydys-Witten equation. In contrast to the former, the latter set is elliptic and the moduli problem can be stated in the standard terms. It has been shown
    by Anderson \cite{Anderson:2012ck} that, the Haydys-Witten equations on a manifold of the type $M_5=M_4\times \BB{R}$ arises as the equation of motion of the twisted maximally supersymmetric gauge theory. But apart from the restriction on the geometry, the supersymmetry algebra of Anderson closes on-shell, and thus the formulation is not suitable if one wants to deform the equations.
    Note that applying large deformations to certain equations has proved fruitful in the study of the Seiberg-Witten equations \cite{Witten_SW,Taubes1994}. One can formulate the deformation more easily in the presence of off-shell supersymmetry.

   The Haydys-Witten equations were proposed by Witten \cite{Witten_FK} in an attempt to understand the Khovanov knot homology from the gauge theory point of view. Independently these equations were also constructed by Haydys \cite{Haydys}, see also \cite{Cherkis:2014xua} for an understanding of these equations in terms of eight dimensional system and octonions. But these will not be the focus of this work.


   The paper is organised as follows:  In section \ref{s-HW-eqs} we go over the Haydys-Witten equation and derive some simple consequences using the energy functional.
   In section \ref{s-5D-YM} we show how one can construct the $N=2$ twisted super Yang-Mills on a K-contact manifold, starting from the formulation \cite{Hosomichi:2012ek} of vector- and hyper-multiplets. We show that after twisting, one can relax
    the   Sasaki-Einstein condition \footnote{The Sasaki-Einstein condition is needed in order to have the Killing spinors. One can certainly consider more general Killing spinor equations with more non-zero background field and  allow for more general manifolds, as was done in \cite{Pan:2013uoa}. However it remains to be done for the case of "squashed" Sasaki-Einstein manifold. For us it is much less laborious to twist the theory first and arrive at the cohomological complex of \cite{Kallen:2012cs}, from which the generalisation is straightforward.}.
   In section \ref{s-summary} we give a summary of the results and discuss some related conjectures.
   To make the paper self-contained we collect some relevant basic definitions of the contact geometry in an appendix.

 \section{The Haydys-Witten equations}\label{s-HW-eqs}

 In this section we discuss the relation between the contact instantons and the Haydys-Witten equations.

 \subsection{Contact instanton}

 We start by recalling some basic facts about the contact instantons. Consider a principal bundle over five
  dimensional manifold $M_5$ with connection $A$ and field strength $F$.
  If $M_5$ admits the contact metric structure $(\gk, \reeb, g)$ then we can define the following PDEs
  \be
  \iota_{\sreeb} F=0~,~~~~F_H^+=0~,\label{contact-instanton}
 \ee
 where $H$ denotes 'horizontal' and $+$ means the self-dual component.  The other notations will be explained shortly or collected in the appendix.
  We impose the anti-self duality condition along the horizontal plane $\xi$ and thus these equations are a
  natural lift of the anti-self-dual instanton equation from 4D to 5D.
  The two equations  (\ref{contact-instanton}) can be  rewritten as one single equivalent equation
  \be
   \star F = - \kappa \wedge F~.\label{contact-instanton-2}
  \ee
  We refer to this equation as anti-self-dual contact instanton. We can also consider the self-dual contact instanton equation,  $\star F = \kappa \wedge F$; but the anti-self dual contact instanton is singled out by the fact that it implies  the Yang-Mills equation $d_A\star F=-d_A(\gk\wedge F)=-d\gk\wedge F=0$.
     The manipulations use the fact that $d \gk$ is horizontal and self-dual, i.e. $\star d\gk = \gk \wedge d\gk$.  In what follows we concentrate on anti-self dual
      contact instanton equation (\ref{contact-instanton}).

 One would like to study the moduli space of the system (\ref{contact-instanton}). However a simple counting reveals that
the equations (\ref{contact-instanton}) are not elliptic and it can be problematic to employ the standard deformation theory for this case.
Although to a degree it is still possible to discuss the moduli space of solutions, see \cite{Baraglia:2014gma}.

\subsection{The Haydys-Witten equations}

In an attempt to understand the Khovanov knot homology from the gauge theory point of view
Witten  \cite{Witten_FK} proposed the system of elliptic equations on five manifold of the form $M_4 \times {\mathbb R}$.
Independently Haydys \cite{Haydys} suggested the system of elliptic equation on $M_5$ with some
additional structure and his equations degenetate to Witten's equation on $M_4 \times {\mathbb R}$. Thus
the Haydys equations can be thought of as covariantisation of Witten's equations.
Below we suggest a covariantisation of Witten's equations, which differs slightly from that of Haydys.
Our version is motivated by vanishing theorems and by the supersymmetry considerations.

Consider the contact metric manifold $M_5$ with $(\gk, \reeb, g)$ and define the following partial
differential equations\footnote{In Haydys' version \cite{Haydys}, the last term in \ref{HW-2} is replaced with $\nabla_{\sreeb}^AB$ which differs from $\iota_{\sreeb} d_A B$
by $\iota_{\sreeb} d_A B-\nabla_{\sreeb}^AB=-[J,B]$ where $[J, B]_{ij} = J_i^{~k} B_{kj} - B_i^{~k} J_{kj}$.}
\bea
&&{\cal A}=\iota_{\sreeb} F - (d_A^\dagger B)^H=0~,\label{HW-1}\\
&& {\cal B}= F_H^+ - \frac{1}{4} B \times B - \frac{1}{2} \iota_{\sreeb} d_A B =0~.\label{HW-2}\eea
To explain the notations, $A$ is the connection and $F$ is its field strength; all other fields transform in the adjoint representation of the gauge group.
The field $B$ is a horizontal self-dual 2-form $\Omega_H^{2+}(M_5)$, see (\ref{hor_self_dual}).
The superscript $(-)^H$ is the projection to the horizontal component: $(-)^H=(1-\gk\iota_{\sreeb})(-)$.  The superscript $+$ means the self-dual component, i.e. $F_H^+
=1/2(F_H+\iota_{\sreeb}\star F_H)$.
Furthermore $d_A = d -i [A, ~]$ and $d_A^\dagger$ is its adjoint with respect to the scalar product
  \be
  (\alpha, \beta) = {\rm Tr} \int\limits_{M_5} \alpha \wedge \star \beta~,~~~~||\alpha||^2 = (\alpha, \alpha)~.\label{scalar-product}
 \ee
 We recall that the space $\Omega_H^{2+}(M_5)$ has a structure of imaginary quaternions, with the product defined as (we use the convention of Witten, see subsection 5.2.5 of \cite{Witten_FK})
 \bea
 (X\times Y)_{mn}=X_{mp}Y_n^{~p}-X_{np}Y^{~p}_m~,~~~X,Y\in\Go_H^{2+}~,\nn
 \eea
 with the indices raised or lowered with the metric. If $X,Y$ are in the adjoint
 \bea
  (X\times Y)^a_{mn}=\frac12f^a_{bc}\big(X^b_{mp}(Y^c)^{~p}_n-X^b_{np}(Y^c)^{~p}_m\big)~,\nn
  \eea
 where $f_{bc}^a$ is the structure constant of the Lie algebra.

 In what follows we refer to the equations (\ref{HW-1})-(\ref{HW-2}) as the Haydys-Witten equations.
  It is straightforward to show that this system of PDEs is elliptic.
  Curiously the  equations (\ref{HW-1})-(\ref{HW-2}) can be combined into a single equation as follows
\bea
\gk\wedge\hat F=-\star\hat F+d_AB~,\label{single_eqn}\eea
where
\bea\hat F=F-\frac14 B\times B~.\nn\eea
To see the equivalence, one first applies $\iota_{\sreeb}$ to both sides and get
\bea \hat F-\gk\iota_{\sreeb}\hat F=-\iota_{\sreeb}\star\hat F+\iota_{\sreeb}d_AB~,\nn
\eea
then use (\ref{used_later}) and that $\go_H=\go-\gk\iota_{\sreeb}\go$ to get
$\hat F_H=-\iota_{\sreeb}\star\hat F_H+\iota_{\sreeb}d_AB$, which is (\ref{HW-2}). To get $(\ref{HW-1})$, we just need to hit
the equation (\ref{single_eqn}) by $\iota_{\sreeb} \star$ and use
the relation
\bea
 (d^{\dag}_AB)^H=\iota_{\sreeb}\star d_AB~,\label{simple_formula}
 \eea
 which can be proved as follows
\bea
d_A^{\dag}B=\star d_A\star B=\star d_A (\gk \wedge B)=\star (d \gk\wedge  B)-\star( \gk \wedge  d_A B)~,\nn
\eea
where we used
$\kappa \wedge B = \star B$ since $B \in \Omega^{2+}_H (M_5)$.
The first term vanishes upon projecting to the horizontal plane, thus
\bea
 (d_A^{\dag}B)^H=-(\star (\gk \wedge d_A B))^H=(\iota_{\sreeb}\star d_A B)^H=\iota_{\sreeb}\star d_A B~,\nn\eea
where we have used (\ref{used_later}). The Haydys-Witten equations (\ref{single_eqn}) are invariant under the following
 scaling symmetry
\bea
 g \rightarrow \lambda^2g~, ~~\gk \rightarrow \lambda\gk~,~~  \reeb \rightarrow \lambda^{-1} \reeb~,~~
 B\rightarrow \lambda B~,~~A\rightarrow A~,\label{scaling-sym}
\eea
 where $\lambda$ is real non-zero number.

   If one has additional vanishing theorem that shows $B=0$, then the equation (\ref{single_eqn})
    collapses to the contact instanton equations (\ref{contact-instanton}). Indeed in coming sections we will prove
     the following:
For a Sasaki manifold with $s+4>0$ pointwise\footnote{Here $s$ is the scalar curvature, this condition essentially saying that the scalar curvature of the transverse K\"ahler metric is greater than 0 everywhere. We will make a remark about how to understand the bound later.}, the $B$ field does not vanish but $B\times B=0$
 and $d_A B=0$ and thus the Haydys-Witten equations degenerate to contact instanton equations. In particular the condition
 $s+4>0$ is true for the Sasaki-Einstein manifolds, with the Reeb vector field close enough to the standard Reeb vector field.
If furthermore, the connection $A$ is an irreducible connection, then $B=0$ altogether.

\subsection{Energy functional}

In this section we restrict ourselves to the closed K-contact manifolds, namely the Reeb vector $\reeb$ is Killing with respect to the compatible metric $g$ (for the definition, see (\ref{can_metric})).

 Now we will investigate  the equations (\ref{HW-1})-(\ref{HW-2})  using the energy functional.
We use the scalar product and the norm defined in (\ref{scalar-product}) under which the three subspaces $\Omega_H^{2+}(M_5)$, $\Omega_H^{2-}(M_5)$ and  $\Omega_V^{2}(M_5)$ are mutually orthogonal.
   Consider the square of the equations (\ref{HW-1})
   \be
   ||{\cal A}||^2= ({\cal A},{\cal A}) = ||\iota_{\sreeb} F||^2 - 2 (\iota_{\sreeb} F, d_A^\dagger B)+
    ||(d_A^\dagger B)^H||^2~,
   \ee
     where for the middle term we can do the following rewriting
   \be
   (\iota_{\sreeb} F, d_A^\dagger B)= (d_A \iota_{\sreeb} F, B) = ({\cal L}_{\sreeb}^A F, B)~,\label{needed_later_I}
   \ee
    where ${\cal L}_{\sreeb}^A$ stands for the gauge covariant version of the Lie derivative, ${\cal L}_{\sreeb}^A = d_A \iota_{\sreeb} + \iota_{\sreeb} d_A $.
    Next consider the square of the equations (\ref{HW-2})
   \be
&& || {\cal B} ||^2 =  || F_H^+ - \frac{1}{4} B \times B||^2 - (  F_H^+ - \frac{1}{4} B \times B, \iota_{\sreeb} d_A B)
 + \frac{1}{4} ||\iota_{\sreeb} d_A B||^2 ~,\nn
   \ee
   where for the middle term
   \be
   (  F_H^+ - \frac{1}{4} B \times B, \iota_{\sreeb} d_A B) = (  F_H^+ - \frac{1}{4} B \times B, {\cal L}_{\sreeb}^A B)~.\nn
   \ee
Therefore  we have
   \be
   (  F_H^+ - \frac{1}{4} B \times B, {\cal L}_{\sreeb}^A B) = (F_H^+, {\cal L}_{\sreeb}^A B) - \frac{1}{4} (B \times B, {\cal L}_{\sreeb}^A B)=
    (F,  {\cal L}_{\sreeb}^A B) - \frac{1}{4} (B \times B, {\cal L}_{\sreeb}^A B) ~,\nn
   \ee
   where the last term will vanish under the current assumption on $\reeb$, as it is a total derivative.  The term $(F,  {\cal L}_{\sreeb}^A B) $ combines with the term $({\cal L}_{\sreeb}^A F, B)$
   in (\ref{needed_later_I}) into a total derivative.  Thus on K-contact $M_5$  we get the following energy functional
   \be
  E= \frac{1}{2} ||{\cal A}||^2 + ||{\cal B}||^2 = \frac{1}{2} ||\iota_{\sreeb} F||+
    \frac{1}{2} ||(d_A^\dagger B)^H||^2
      +  \frac{1}{4} ||\iota_{\sreeb} d_A B||^2 +
  || F_H^+ - \frac{1}{4} B \times B ||^2~. \label{energy}
   \ee
Now we arrive at the following system of equations on a closed K-contact manifold
\begin{equation}\label{HW_eqn}
\begin{split}
&\iota_{\sreeb} F=0~,\\
&F_H^+ - \frac{1}{4} B \times B=0~,\\
&(d_A^\dagger B)^H=0~,\\
&\iota_{\sreeb} d_A B=0~.
\end{split}
\end{equation}
   The last two equations $(d_A^\dagger B)^H=0$ and $\iota_{\sreeb} d_A B=0$ can be combined into a single equation $d_A B=0$ as follows:
   Using the relation\footnote{Note that this is a stronger version of (\ref{simple_formula}), the proof of this relation requires K-contact manifolds while the proof of (\ref{simple_formula}) does not.}
   \bea
   \star d_AB=\L^A_{\sreeb}B+\gk d^{\dag}_AB~,\nn
   \eea
   one has $\iota_{\sreeb}\star d_AB=\iota_{\sreeb}(\gk d^{\dag}_AB)=(d^{\dag}_AB)^H=0$.
   Then $\iota_{\sreeb}\star d_AB=\iota_{\sreeb} d_A B=0$ implies $d_AB=0$. Conversely, if $d_AB=0$ we get
   $\L^A_{\sreeb}B+\gk d^{\dag}_AB=0$, but $\L^A_{\sreeb}B=\{\iota_{\sreeb},d_A\}B=0$ and so $\gk d^{\dag}_AB=0$ giving $(d^{\dag}_AB)^H=0$.
   Therefore we conclude that on K-contact manifold the Haydys-Witten equations (\ref{HW-1})-(\ref{HW-2})  collapse to
    the following system of equations
   \begin{equation}\label{perp_dk}
   \begin{split}
   & \iota_{\sreeb} F=0~, \\
   &  F_H^+ - \frac{1}{4} B \times B=0~, \\
   & d_A B=0~.\\
   \end{split}
   \end{equation}

\subsection{General K-contact case}
The above conditions on $B$ are very restrictive, which makes one wonder if one can establish $B=0$ altogether, and in particular, under which geometry do the Haydys-Witten equations (\ref{HW-1})-(\ref{HW-2}) collapse to the contact instanton equations (\ref{contact-instanton-2}). The following discussion
is analogous to \cite{Vafa:1994tf}, and is based on the repeated use of the Weizenbock formula.

Let $B,~C\in\Go_H^{2+} (M_5)$, then
\bea
(d_A^{\dag}C,d_A^{\dag}B)&=&-\frac12(\L^A_{\sreeb}C,\L^A_{\sreeb}B)+\frac12(J\cdotp C,J\cdotp B)+(F,C\times B)\nn\\
&&+\Tr\int\limits_{M_5} \left ( \frac{1}{4}\nabla C\cdotp\nabla B-\frac14C_{ij}X^{ijkl} B_{kl} \right ) {\rm vol}_g~,\label{Weizenbock}
\eea
where $J\cdotp B=J^{pq}B_{pq}$, $\nabla C\cdotp\nabla B=(\nabla_i C_{jk})(\nabla^iB^{jk})$ with $\nabla$ containing both the Levi-Civita and gauge connections. The second term in the second line is a quadratic form defined on the horizontal self-dual 2-forms, with
\bea
X_{pqrs}=R_{pqrs}-\frac12(Ric\bar\wedge g)_{pqrs}=W_{pqrs}-\frac1{6}(Ric\bar\wedge g)_{pqrs}-\frac{s}{24}(g\bar\wedge g)_{pqrs}~,\label{B_bilinear}\eea
where $R$ is the Riemann tensor, $Ric$ is the Ricci tensor, $s$ is the Ricci scalar, $W$ is the Weyl-tensor whose definitions and
   that of $\bar\wedge$ are collected in the appendix.


Using the Weizenbock formula (\ref{Weizenbock}) and the relation $d_A^{\dag}B=(d_A^{\dag}B)^H-\gk (J\cdotp B)$
the energy functional (\ref{energy})
 can be rewritten as follows
\bea
E &=& \frac{1}{2} ||\iota_{\sreeb} F||^2 + ||F_H^+||^2 + \frac{1}{16}||B \times B||^2+\frac18\Tr\int\limits_{M_5}  ( \nabla B\cdotp\nabla B)~ {\rm vol}_g\nn\\
 &&-\frac14||J\cdotp B ||^2-\frac18\Tr\int\limits_{M_5} ( B_{ij}X^{ijkl}B_{kl})~ {\rm vol}_g~,\label{numbered}\eea
where one notices that the curvature term $F\wedge\star(B\times B)$ is cancelled.
 In the energy functional only the second line is not positive definite and thus the central question is to establish when are these terms
  positive definite.

Note that the expression of $X$ in (\ref{B_bilinear}) has no gauge connection, so its positivity can be answered regardless of the gauge theory.
For a generic K-contact manifold, we follow the treatment of \cite{Vafa:1994tf} and consider the Weizenbock formula (\ref{Weizenbock}) again, without the gauge theory part
\bea
||(d^{\dag}w)^H||^2=-\frac12||{\cal L}_{\sreeb}w||^2-\frac12||J\cdotp w||^2-\frac14 \int\limits_{M_5} (w_{ij}X^{ijkl}w_{kl})~{\rm vol}_g+\frac14 \int\limits_{M_5} \nabla w\cdotp\nabla w ~{\rm vol}_g~,\nn
\eea
with $w\in\Go_H^{2+}$. Restricting to $w$ such that $\L_{\sreeb}w=0$, i.e. $w$ is basic (see below), we have
\bea
||(d^{\dag}w)^H||^2=-\frac12||J\cdotp w||^2-\frac14\int\limits_{M_5} (w_{ij}X^{ijkl}w_{kl}) {\rm vol}_g+\frac14 \int\limits_{M_5} \nabla w\cdotp\nabla w~{\rm vol}_g  ~.\label{numbered_I}\eea
From here one can show that the first two terms of the above (or the second line of (\ref{numbered})) will \emph{never} be positive definite.

To see this we need some facts about the \emph{basic cohomology} associated to the Reeb foliation, which are collected toward the end of the appendix.
We assume that $\go\in\Go^2(M_5)$ such that $\iota_{\sreeb}\go=0=\L_{\sreeb}\go$, i.e. they are basic forms.
We also assume that $M_5$ is simply connected, then the short exact sequence (\ref{exact_sequence}) is applicable. So the basic cohomology group $H^2_B$ consists of the entire $H^2$ plus an extra class generated by $d\gk$. Every element  $\go\in H_B^2$ has a harmonic representative such that $d\go=(d^{\dag}\go)^H=0$. In particular, it is easy to check that $d\gk$ is harmonic. Now if the first two terms in (\ref{numbered_I}) were positive, it would imply that $d\gk$ is covariantly constant, which is not true
\bea
\reeb^j\nabla_i(d\gk)_{jk}=\nabla_i(\reeb^j(d\gk)_{jk})-(\nabla_i\reeb^j)(d\gk)_{jk}=2(g_{ik}-\reeb_i\reeb_k)\neq0~.\nn
\eea
In view of this the question of the vanishing of $B$ is left unanswered without further assumptions on the
 geometry. To proceed, one notices that it is the class $[d\gk]$ that spoils the positivity argument, so one
 needs to apply the Weizenbock formula away from it.

%
%
%
%
%

\subsection{Sasaki case}
We now restrict the discussion to the Sasaki manifolds (see the appendix  for the definitions).
As a first remark, the three components of $\Go_H^{2+}(M_5)$ have the structure of the imaginary quaternions, we will denote by $B_3$ the component of $B$ that is proportional to $J$ or $d\gk$, and $B_1,~B_2$ the components that are orthogonal to $J$. In fact, $B_1\pm iB_2$ will be of type (2,0) or (0,2) with respect to the horizontal complex structure. Furthermore $d_A^{\dag}B_{1,2}$ are automatically horizontal. One does not need the vanishing of all three $B$'s for the Haydys-Witten equations to collapse to instanton equations, the vanishing of two will ensure $B\times B=0$ and will therefore be sufficient. This is what will happen eventually.

We first establish some orthogonality statements.
\begin{lemma}
Let $B_3=Jf$, with $f$ an adjoint scalar, and $C\in\Go_H^{2+}(M_5)$ be orthogonal to $J$ and in the adjoint representation too. Then
\bea \nabla B_3\cdotp\nabla C=0\nn\eea
on a Sasaki manifold.
\end{lemma}
\begin{proof}
A direct computation shows
  \bea \nabla_i (J_{jk}f)\nabla^iC^{jk}&=&(f\nabla_iJ_{jk}+J_{jk}\nabla_if)\nabla^iC^{jk}=\big(-2f\reeb_jg_{ik}+J_{jk}\nabla_if\big)\nabla^iC^{jk}\nn\\
  &=&2f(\nabla^i\reeb_j)g_{ik}C^{jk}-(\nabla^iJ_{jk})\nabla_ifC^{jk}\nn\\
  &=&-2fJ^i_{~j}g_{ik}C^{jk}-\nabla_if(-2\reeb_j\gd^i_k)C^{jk}=0~,\nn\eea
  where in the second step one needs the integrability condition (\ref{integrable_J})\qed
\end{proof}
\begin{lemma}
Let $B_3$ and $C$ be as above, then
\bea
B_3^{ij}X_{ijkl}C^{kl}=0\nn
\eea
on a Sasaki manifold.
\end{lemma}
\begin{proof}
We recall that
\bea B_3^{ij}X_{ijkl}C^{kl}=B_3^{ij}(R_{ijkl}-2R_{pi~k}^{~~p}g_{jl})C^{kl}~.\nn\eea
By using (\ref{Riemann_20}) one sees that the first term vanishes, since $C$ is of type (0,2) or (2,0).

Using (\ref{Riemann_11}), the second term gives $2(c_1)_{kl}C^{kl}$, where $c_1$ is defined after (\ref{Riemann_11}), and is the analogue of the first Chern class in the K\"ahler case.
A direct calculation then shows  that $c_1$ is (1,1). Hence the second term also vanishes\qed
\end{proof}
It is also quite clear that
\bea (\L^A_{\sreeb}B_3,\L^A_{\sreeb}C)=0~.\label{numbered_II}
\eea

Now we apply the Weizenbock formula to the right hand side of
\bea
||(d_A^{\dag}B)^H||^2-||(d_A^{\dag}B_3)^H||^2=||d_A^{\dag}\tilde B||^2+2(d_A^{\dag}\tilde B,(d_A^{\dag}B_3)^H)~,\nn
\eea
where we use $\tilde B$ to denote those components of $B$ that are orthogonal to $J$, and we can remove the projector ${}^H$ from the second term
\bea
 ||d_A^{\dag}\tilde B||^2+2(d_A^{\dag}\tilde B,d_A^{\dag}B_3)&=&\frac14\Tr\int (\nabla\tilde B\cdotp\nabla \tilde B-\tilde BX\tilde B)~\textrm{vol}_g\nn\\
&&+(F,\tilde B\times \tilde B)+2(F,\tilde B\times B_3)-\frac12||\L_{\sreeb}\tilde B||^2~,\nn
\eea
where we have used the two previous lemmas and (\ref{numbered_II}) to eliminate some of the cross terms. We also notice that since
$B_3\times B_3=0$,
\bea
 (F,\tilde B\times \tilde B)+2(F,\tilde B\times B_3)=(F,B\times B)~.\nn
 \eea
After this manipulation, the energy functional becomes
\bea
E&=&\frac{1}{2}||\iota_{\sreeb} F||+\frac{1}{2} ||(d_A^\dagger B_3)^H||^2+\frac{1}{4} ||\L^A_{\sreeb}B_3||^2+||F_H^+||^2+\frac{1}{16}||B \times B ||^2\nn\\
&&+\frac18\Tr\int (\nabla \tilde B\cdotp\nabla \tilde B-\tilde BX\tilde B)~\textrm{vol}_g~.\nn\eea
The key issue is still the negativity of the last term, which we can now compute explicitly.
Now let $Z$ be horizontal of type (2,0) and $\bar Z$ of type (0,2), and then we have
\bea
 Z^{mn}R_{mnpq}\bar Z^{pq}~{\rm vol}_g =4Z\wedge \star\bar Z~,~~~~
Z^{mn}(R\bar\wedge g)_{mnpq}\bar Z^{pq}~{\rm vol}_g=2(s-4)Z\wedge \star\bar Z~.\label{two_relations}
\eea
The first relation follows immediately from (\ref{Riemann_20}). For the second one, a direct calculation shows
\bea
&& Z^{mn}\bar Z^{pq}(R\bar\wedge g)_{mnpq}=-4Z^{mn}\bar Z^{pq}J_p^{~r}J_q^{~s}R_{mr}g_{ns}\nn\\
&&\hspace{3.4cm}=-4Z^{mn}\bar Z^{pq}J_p^{~r}J_{qn}R_{mr}=-4 Z^{mn}\bar Z^{pq}(c_1)_{mp}J_{nq}~,\nn\eea
where in the first step we have inserted two $J$'s and an accompanying $-$ sign because $\bar Z$ is of type (0,2). Next one uses the fact that $c_1$ is of type (1,1)
\bea
 Z^{mn}\bar Z^{pq}(c_1)_{mp}J_{nq}=\frac1{16}Z^{[mn}\bar Z^{pq]}(c_1)_{mp}J_{nq}=-\frac1{24}(Z\wedge\bar Z)^{mnpq}
 (c_1\wedge J)_{mnpq}~.\nn\eea
Since $c_1\wedge J$ is a horizontal 4-form, it has to be proportional to $J\wedge J$
\bea c_1\wedge J=\frac14(s-4)J\wedge J~,\nn\eea
where $s$ is the Ricci scalar. Then we arrive at
\bea Z^{mn}(R\bar\wedge g)_{mnpq}\bar Z^{pq}~\textrm{vol}_g=2(s-4)Z\wedge\star \bar Z~.\nn\eea
And finally the quadratic form goes to
\bea
\Tr\int\limits_{M_5}(\tilde B_{ij}X^{ijkl}\tilde B_{kl})~{\rm vol}_g=\Tr\int\limits_{M_5}(8-s)\tilde B\wedge\star\tilde B~.\label{B12XB12}
\eea
For completeness we will also give $B_3XB_3$
\bea B_3^{mn}X_{mnpq}B_3^{pq}~\textrm{vol}_g=-12B_3\star B_3~.\label{B3XB3}\eea

As an example of the above calculation, take $S^5$ with the standard metric. Use the presentation of $X$ as in (\ref{B_bilinear}), then the Weyl tensor vanishes since the metric is conformally flat, also $Ric=4g$ and $s=20$. So the other two terms in $X$ give
\bea
X=0-\frac16\cdotp 4g\bar\wedge g-\frac{20}{24}g\bar\wedge g=-\frac32 g\bar\wedge g~,\nn\eea
then $BXB~{\rm vol}_g =-12(B\wedge \star B)$. This agrees with (\ref{B12XB12}) and (\ref{B3XB3}).

To summarise the above computation, we come to the following conclusion.
If a Sasaki manifold has $s>8$ point wise, then $\tilde B=0$. When this happens, the Haydys-Witten
equations (\ref{perp_dk}) degenerate to the contact instanton equations.

 However we can improve this bound significantly. Unlike the symmetry \eqref{scaling-sym}, the Sasaki-metric \eqref{can_metric} does posses another symmetry. If we apply the so-called D-homothetic transformations with a parameter $a >0$
 to the Sasaki structure
\bea
\gk\to a\gk~,~~~~\reeb\to a^{-1}\reeb~,~~~~g \to a g + (a^2-a) \gk \otimes \gk \label{scaling-sym_I}
\eea
with the complex structure $J$ is unchanged then we get another Sasaki structure (see \cite{BoyerGalicki} for more explanation).
 This transformation scales the transverse metric by $a$ and vertical one by $a^2$. This is not a symmetry of the
Haydys-Witten equations (\ref{HW-1}) and (\ref{HW-2}), but preserves the \emph{decoupled} Haydys-Witten
 equations \eqref{perp_dk} if we also let $B\to a^{1/2}B$.

Under this symmetry the scalar curvature changes as
\bea
 s\to a^{-1}(s+4)-4~.\label{scaling_curvature}\eea
Thus provided $s>-4$ everywhere, one can use this scaling to make $s>8$ everywhere and making the vanishing result earlier applicable. To summarise
\begin{theorem}
  For a compact Sasaki 5-manifold, if $s>-4$ everywhere, then the Haydys-Witten
   equations \eqref{perp_dk} degenerate into the contact instanton equations
  \bea\iota_{\sreeb} F=F_H^+=d_A B=0~,\nn\eea
  and $B$ is proportional to $d\gk$.
\end{theorem}
\begin{remark}
  In fact it is possible to modify the proof above slightly and get straight to the bound $s>-4$ without resorting to the symmetry \eqref{scaling-sym_I}. The main trick is to further split the term $\nabla \tilde B\cdotp\nabla \tilde B$ into mutually orthogonal components
  \bea \nabla \tilde B\cdotp\nabla \tilde B=(\nabla \tilde B)_H\cdotp(\nabla \tilde B)_H+12\tilde B*\tilde B,\nn\eea
  where $(~)_H$ means the horizontal component, and also on the rhs terms proportional to $L_{\sreeb}\tilde B=0$ are dropped since they must vanish by the HW equation.
  Combining this with \eqref{B12XB12} one modifies $s-8\to s-8+12=s+4$ which is now homogeneous under \eqref{scaling-sym_I}: $s+4\to a^{-1}(s+4)$ \footnote{We would like to thank the referee for pointing out the symmetry \eqref{scaling-sym_I} that prompted us to improve our proof and get a stronger vanishing result.}.
\end{remark}

Specifying now to the Sasaki-Einstein manifolds, whose partition function were the main subject of study in \cite{Qiu:2013aga,Qiu:2014oqa}, one has $Ric=4g$, $s=20>-4$ satisfying the bound, thus our vanishing theorem applies. Surely, we have performed the calculation assuming the Sasaki-Einstein metric, and each such metric usually comes with a fixed Reeb vector field (in the toric case, see \cite{Martelli:2005tp}). Take for instance $S^5$, the standard round Sasaki-Einstein metric requires the Reeb vector field to be the regular one, namely it is the fibre of the Hopf fibration $S^1\to S^5\to\BB{CP}^2$, and every Reeb orbit is closed. However, in many cases it is desirable to deform the Reeb vector field slightly while keeping the metric Sasaki. In this case the positive definiteness can also be argued if one recalls that positive definiteness is a condition stable under small perturbations. For general Sasaki-Einstein manifolds, even when the associated Reeb vector field is already irregular, one would still like the freedom to deform the Reeb vector field, so as to study the dependence of the partition function on the contact geometry, see \cite{Qiu:2014oqa}.

We conclude that for deformed Reeb vector field we have $\tilde B=0$ provided the Reeb vector field
is sufficiently close to the one corresponding to Sasaki-Einstein metric and the deformed metric is kept Sasaki.

One may wonder when  $B_3=0$ on a Sasaki manifold.
Due to the Sasaki structure one can split the differential $d^B$ into $d^B=\partial^B+\bar\partial^B$, with
\bea
\bar\partial^B:~~\Go_H^{p,q}\to\Go_H^{p,q+1}~;~~~\partial^B:~~\Go_H^{p,q}\to\Go_H^{p+1,q}~,\nn
\eea
just as in the K\"ahler manifold case. In fact, one has also the well-known relation between the Laplacian of $d^B$ and $\bar\partial^B$
\bea
\Gd^B=\{d^B,(d^B)^{\dag}\}~,~~~\Gd^B_{\bar\partial^B}=\{\bar\partial^B,(\bar\partial^B)^{\dag}\}~;~~~\Gd^B=2\Gd_{\bar\partial^B}^B=2\Gd_{\partial^B}^B~.\nn
\eea
The incorporation of the gauge connection does not change the conclusion.

Our $B$ field satisfies $d_AB=(d_A^{\dag}B)^H=0$, i.e. $B$ is  harmonic, then the last relation above shows that $B$ is also closed
 with respect to $\bar\partial^B$ or $\partial^B$, where the gauge connection is included but not written. Then a simple degree consideration shows that $d_AB_3=0$ all by itself. In particular, we write $B=J\,f$ for some adjoint scalar $f$, then
\bea
0=d_AB_3=d_A(Jf)=(dJ)f+Jd_Af=Jd_Af~.\nn
\eea
This implies that $d_Af=0$ since $J$ is horizontal and non-degenerate in the horizontal plane.

Now if we assume that the connection is \emph{irreducible}. Then $d_Af=0$ would imply that $f=0$ and hence $B_3=0$.

\section{Twisted 5D $N=2$ Yang-Mils theory}\label{s-5D-YM}
In this section we construct  the $N=2$ cohomological complex ($N=2$ twisted supersymmetry)
by combining the $N=1$ complex of vector- and hyper-multiplets. The localisation locus associated to the $N=2$ theory is naturally (\ref{HW_eqn}), and this is how we came to suggest our version of the covariantisation of the Haydys-Witten equation in the first place.

Let us outline the main steps: The $N=1$ cohomological complex for vector-multiplet is written in terms of differential forms on $M_5$ and
it requires $M_5$ to be K-contact.  The $N=1$ cohomological complex of the hyper-multiplet (\ref{susy_hyper_twist}) is written in terms of spinors, and requires the Sasaki-Einstein structure. One can pick a specific spin representation and rewrite the hyper-multiplet trasnformations in terms of differential forms. After this step, one finds that the cohomological complex for hyper-multiplet  is valid for any K-contact manifold. Next
  we take the cohomological complex of a vector-multiplet and hyper-multiplet in adjoint representation,
  and denote the supersymmetry $\gd_1$.
Now one tries to find a $U(1)$ symmetry that would mix the vector-complex with the hyper-complex, by taking the commutator of $\gd_1$ with this $U(1)$, one necessarily gets a new supersymmetry $\gd_2$. We will combine $\gd_{1,2}$ into a complex transformation $\delta$.

Let us summarise the main result of this section.
On any $K$-contact manifold $M_5$ we can define the following $N=2$ twisted supersymmetry
\bea
\gd_{\ep}A&=&\bar\ep\Psi+\ep\bar\Psi~,\nn\\
\gd_{\ep}\gs&=&-2\bar\ep\iota_{\sreeb}\bar\Psi~,\nn\\
\gd_{\ep}\bar\gs&=&-2\ep\iota_{\sreeb}\Psi~,\nn\\
\gd_{\ep}B&=& -\bar\ep\chi+\ep\bar\chi~,\nn\\
\gd_{\ep}H&=&-i{\cal L}_{\sreeb}^A(\bar\ep\chi+\ep\bar\chi)- \ep [\gs, \chi] -\bar\ep [\bar\gs, \bar\chi]- \iota_{\sreeb}[B,\bar\ep\Psi-\ep\bar\Psi]~,\label{final-N2complex}\\
\gd_{\ep}\tilde{\cal F}&=&\ep\big(i\iota_{\sreeb} d_A\bar\Psi-[\gs,\Psi^H\big])+\bar\ep(-i\iota_{\sreeb}d_A\Psi+[\bar\gs, \bar\Psi^H])~,\nn\\
\gd_{\ep}\Psi&=&\ep(2\tilde{\cal F}-2i\iota_{\sreeb}F+\gk [\gs, \bar\gs])+2i\bar\ep d_A\bar\gs~,\nn\\
\gd_{\ep}\bar\Psi&=&\bar\ep(-2\tilde{\cal F}-2i\iota_{\sreeb}F -\gk [\gs, \bar\gs])+2i\ep d_A\gs~,\nn\\
\gd_{\ep}\chi&=&2\ep(H- i{\cal L}_{\sreeb}^AB)+\bar\ep [\bar\gs, B]~,\nn\\
\gd_{\ep}\bar\chi&=&2\bar\ep(H+i {\cal L}_{\sreeb}^AB)-\ep [\gs, B]~.\nn
\eea
Here $\ep=1/2(\ep_1+i\ep_2)$, $\bar\ep=1/2(\ep_1-i\ep_2)$ are the parameters for the susy transformation, $A$ is the gauge connection, $F$ is its field strength, $\sigma$ is a complex scalar in the adjoint representation. Furthermore
 $B,~H\in\Omega_H^{2+}(M_5)$ are real adjoint, while $\tilde{\cal F}$ is a real adjoint horizontal 1-form. The fermionic fields are: $\Psi$, a complex 1-form in the adjoint; and $\chi$, a complex form belonging to $\Omega_H^{2+}(M_5)$ in the adjoint. The superscript $H$ on $\Psi^H$ means
    the horizontal component and the pairing $\bra\Psi,B\ket$ is defined as $\bra\Psi,B\ket_l = [\Psi_j, B_{kl}] g^{jk}$.  The $U(1)$ charge is
     allocated as follows: $A$, $B$, $H$ and ${\cal F}$ have charge 0, while $\ep$, $\Psi$ and $\chi$ (resp. $\bar\ep$, $\bar\Psi$ and $\bar\chi$) have charge $+1$ (resp. $-1$),
      and finally $\bar{\gs}$ has charge $+2$ ($\gs$ has charge $-2$).
       The supersymmetry (\ref{final-N2complex}) satisfies the following $N=2$ algebra:
       \bea
        \{\delta_{\ep},\gd_{\eta}\}=-4\eta\ep G_{\bar\gs}-4\bar\eta\bar\ep G_{\gs}-4i(\eta\bar\ep+\ep\bar\eta) (\L_{\sreeb}-iG_{\iota_{\sreeb}A})~,\nn\eea
where $G_{\phi}$ is the gauge transformation with parameter $\phi$ defined as
\bea
&&G_{\phi}A=-id_A \phi~,\nn\\
&&G_{\phi}-=[\phi,-]~,\label{gauge_transformation}
\eea
where the second line is for all other field in the adjoint. Next we come to the derivation of this result.

\subsection{The $N=1$ complex}
The $N=1$ cohomological complex for the vector-multiplet was written down in \cite{Kallen:2012cs} for any K-contact manifold
\bea
&&\begin{array}{ll}
  \delta A = \Psi~, & \gd \Psi = -i\iota_{\sreeb}F+id_A \gs~, \\
  \gd \chi = H~, & \gd H=-i{\cal L}^A_{\sreeb}\chi-[\gs,\chi]~, \\
  \gd \gs =-\iota_{\sreeb}\Psi~, &
\end{array}\label{susy_vect_twist}
\eea
 and its relation to the $N=1$ supersymmetry has been discussed in details in \cite{Kallen:2012va}.
In (\ref{susy_vect_twist}) $A$ is the gauge field, $d_A=d-i[A, ~]$ is the gauge covariant derivative. All other fields are in the adjoint; they include: $\gs$  a \emph{real} scalar, $\Psi$ an odd \emph{real} 1-form, $\chi$ (resp. $H$) an odd (resp. even) \emph{real} horizontal self-dual 2-form, i.e. $\iota_{\sreeb}H=0$, ${\star}_{\sreeb}H=\iota_{\sreeb}{\star}H=H$.
The closure of $\gd$ reads
\bea
 \gd^2=-i{\cal L}_{\sreeb}-G_{\gs+\iota_{\sreeb}A}~,\label{closure_vector}
\eea
where $G_{\phi}$ is the gauge transformation defined in (\ref{gauge_transformation}).

The $N=1$ cohomological complex of the hyper-multiplet is written in terms of spinors and thus for the moment we assume the Sasaki-Einstein condition. Later we will show how to relax this restriction on the geometry.
We define a projector $P_{\pm}=\frac{1}{2}(1\pm \gamma_5)$ with $\gc_5=-\reeb\cdotp\Gc$, namely the nowhere vanishing $\reeb$ gives us an operator $\gc_5$ that
can be used to split the spin bundle according to the chirality.  The horizontal complex structure
$J$ is related to the Reeb as (indices are raised and lowered with the metric implicitly)
\bea
 J_{pq}=-\frac12\nabla_{[p}\reeb_{q]}~.\nn
\eea
By definition $J$ is horizontal and self-dual. The complex for the hyper-multiplet with fields $(q, \psi, {\cal F})$
in a general representation of the Lie algebra reads as follows \cite{Kallen:2012va}
\bea
&& \delta q = iP_+\psi~,\nn\\
&& \delta \psi =-\frac{1}{4r}J_{pq}(\Gc^{pq}q)+(\slashed{D}+i\gs) q+{\cal F}~,\label{susy_hyper_twist}\\
&& \delta {\cal F} = -iP_-\slashed{D}\psi-\gs P_-\psi+i\Psi^m(\Gc_m+\reeb_m)q~,\nn\eea
where $q$ is an even spinor with $P_-q=0$, and ${\cal F}$ is an auxiliary even spinor with $P_+{\cal F}=0$.
The odd $\psi$ has both chiralities under $\gc_5$ and we write $\psi_{\pm}$ for $P_{\pm}\psi$.
Here  $\slashed{D}$ is the Dirac operator.

In the following, we need to combine the two complexes (\ref{susy_vect_twist}) and (\ref{susy_hyper_twist}).
  For this we will need to make a change of variable for the hyper-complex  and to rewrite it terms of differential forms.
 The following discussion will be technical and the reader may consult the final answer the  table \ref{susy_combo} for
  supersymmetry transformations.

We will make use of one of the two Killing spinors satisfying the Killing equation (the other Killing spinor will have the opposite sign in the equation below)
\bea
D_m\xi=-\frac{i}{2}\Gc_m\xi~,\nn
\eea
which satisfies
\bea
 &&P_-\xi=0~,~~~\big((JX)\cdotp\Gc-\frac{i}{2}(\slashed{X}+\slashed{\reeb}\slashed{ X})\big)\xi=0~,\label{subbundle}\\
&&\hspace{2cm}\slashed{J}\xi=-4i\xi~. \label{spin_form_degree}\eea
 The convention of the gamma matrices are collected at the end of the appendix.
By using $\xi$, all the other spinors in the hyper-complex can be reduced into horizontal $(0,i)$-forms. Another way of saying this is that we use the horizontal $(0,i)$-forms to build a spin representation. With this representation $J$ has the properties
\bea
 \slashed{J}=-4i(1-\deg)~,~~~J_{\bar i}^{~\bar j}=i\gd^{\bar j}_{\bar i}~,\nn
 \eea
and $\xi$ can be thought of as the (0,0)-form: $\xi\sim 1$.

As an example of this rewriting
\bea
&&q\To \slashed{B}\xi+f\xi~,~~~\psi_+\To \slashed{\Gs}\xi+\gl\xi~,\nn\\
&&\psi_-\To \psi_{-p}\Gc^p\xi~,~~~{\cal F}\To {\cal F}_p\Gc^p\xi~,\label{spinor_diff_form}
\eea
where $B$ and $\Gs$ are horizontal (0,2)-forms and $f,\,\gl$ are 0-forms. In the second line we used the same symbols $\psi_-,\,{\cal F}$ for the spinors $\psi_-,\,{\cal F}$ and the (0,1)-forms they reduce to. The supersymmetry rule reads simply $\gd B=i\Gs$ and $\gd f=i\gl$.
To obtain $\gd\Gs$ or $\gd\gl$, we need to rewrite $\gd\psi$ as
\bea &&\gd\psi
=(-{\cal L}^A_{\sreeb}B-\frac{3i}{2}B+i\gs B)_{pq}\Gc^{pq}\xi\nn\\
&&\hspace{4cm}+(-4d_A^{\dag}B+(d_A f)^{(0,1)})_p\Gc^p\xi
+(-{\cal L}^A_{\sreeb}f-\frac{3i}{2}f+i\gs f)\xi+{\cal F}~,\nn\eea
where $(\,)^{0,1}$ stands for the horizontal (0,1)-component. From these, we get the expected results
\bea
&&\gd \Gs=-{\cal L}^A_{\sreeb}B-\frac{3i}{2}B+i\gs B~, \label{psi_Sigma}\\
&&\gd \gl=-{\cal L}^A_{\sreeb}f-\frac{3i}{2}f+i\gs f~.\label{psi_lambda}
\eea
We can do likewise for ${\cal F}$ and get
\bea \gd{\cal F}=-i{\cal L}_{\sreeb}^A\psi_-+\frac{3}{2}\psi_--((d_A\gl)^{0,1}\cdotp\Gc)\xi+4i(d_A^{\dag}\Gs\cdotp\Gc)\xi-\gs \psi_-+4i\langle \Psi, B\rangle+i\Psi^{0,1}f~,\nn
\eea
 where $\langle \Psi, B\rangle_i = \Psi_l B_{ki} g^{lk}$ and  with Lie algebra indices it is understood as matrix multiplication
  (later for the adjoint hyper it will be simply the commutator).
To summarise we have the supersymmetry which closes as $\gd^2=-i\L_{\sreeb}-G_{\gs+\iota_{\sreeb}A}+3/2$. The extra shift of 3/2 compared to the vector case (\ref{closure_vector}) has some deep geometrical and physical meanings \cite{Schmude:2014lfa}, and it is the reason that in the presence of $N=1$ supersymmetry, an adjoint hyper will not cancel completely the vector contribution. But as we are aiming for an $N=2$ complex, we need to modify the supersymmetry
  rules for the hyper-complex to make the closure property identical for the vector- and hyper-complex.
  In the end the modified hyper-complex reads
\begin{equation}\label{susy_hyper_twisted}
\begin{split}
&\gd f=i\gl,~~~~\gd B=i\Gs~,\\
&\gd{\cal F}=-i\big({\cal L}_{\sreeb}^A\psi_-+(d_A \gl)^{0,1}-4d_A^{\dag}\Gs\big)-\gs \psi_-+4i \langle \Psi,
  B \rangle +i\Psi^{0,1}f~,\\
&\gd \Gs=- {\cal L}^A_{\sreeb}B+i\gs B~,~~~~\gd \gl=-{\cal L}^A_{\sreeb}f+i\gs f~,\\
&\gd\psi_-=-4d_A^{\dag}B+(d_Af)^{(0,1)}+{\cal F}~.
\end{split}
\end{equation}
where the superscript $(0,1)$ means to take the horizontal and $(0,1)$ component of the 1-form in question.

 From now on we assume that all fields from hyper-multiplet are in adjoint representation and we write
  the commutators explicitly in our formulas.
To combine the above with (\ref{susy_vect_twist}), we need to split the fields in (\ref{susy_hyper_twisted}) into real and imaginary parts. Let us sketch the fate of the various fields.
First, the real and imaginary parts of $B$ are related by $J$, and the real part gives us two out of three horizontal self-dual 2-forms, whereas the third one, which is propositional to $J$, will be given by $\im f J$. The fields $\re B$ and $\im fJ$ altogether will again be called $B$. At the same time $\re f$ will become the imaginary part of $\gs$ in the vector-complex. Then it is appropriate to combine $J\re\gl,\im\Gs$ with $\chi$ into a complex horizonal self-dual 2-form.
Second, from the 1-form $\psi_-$ (automatically (0,1)), it suffices to take its imaginary part and combine it with $\Psi$ from the vector part to form a complex odd 1-form. But as $\psi_-$ is always horizontal, the missing vertical component will be supplied by $\im \gl$. As for ${\cal F}$, we will continue to denote by ${\cal F}$ its real part, and as always its imaginary part is not independent. To summarise, we have now four real fields $A,\,H$ and ${\cal F},\,B$ and a number of complex fields $\Psi,\chi,\gs$. In the following table of susy rules, we start to denote $\gd$ as $\gd_1$ to differentiate it with the new susy that we will get shortly. For the complex fields, we use the subscript ${}_{1,2}$ denote their real and imaginary parts
\bea
&&\gd_1A=\Psi_1~,\nn\\
&&\gd_1 B=-\chi_2,~~~\gd_1H=-i{\cal L}_{\sreeb}^A\chi_1 -[\gs_1, \chi_1]~,\nn\\
&&\gd_1{\cal F}=\iota_{\sreeb}d_A\Psi_2-4(d_A^{\dag}\chi_2)^H-i\gs_1\Psi_2^H+4i \langle \Psi_1, B \rangle +i[\Psi_1^H,\gs_2]~,\nn\\
&&\gd_1\gs_1=-\iota_{\sreeb}\Psi_1~,~~~\gd_1\gs_2=\iota_{\sreeb}\Psi_2~,\nn\\
&&\gd_1\Psi_1=-i\iota_{\sreeb}F+id_A\gs_1~,~~\gd_1\Psi_2=4i(d_A^{\dag} B)^H-id_A\gs_2-i{\cal F}-\gk[\gs_1,\gs_2]~,\nn\\
&&\gd_1\chi_1=H~,~~~\gd_1\chi_2=i{\cal L}^A_{\sreeb} B+[\gs_1 ,B]~,\label{susy_one}
\eea
where $(\,)^H$ stands for the horizontal component.

\subsection{A $U(1)$ symmetry and the $N=2$ complex}

To obtain the second supersymmetry consider the assignment of R-charges
\bea \begin{array}{c|c|c|c}
        R & 0 & 1 & 2 \\
        \hline
        \textrm{fields} & A,B,H,{\cal F} & \Psi,\chi &\bar\gs\\
       \end{array}~,\label{list_R_charge}\eea
we will use $\rho$ to denote the infinitesimal generator of the $R$-rotation. It acts as
\bea
\rho\chi=i\chi~;~~~\rho\Psi=i\Psi~;~~~\rho\gs=-2i\gs \label{def_rho}
\eea
and zero on all other fields.

This $R$-charge can be traced back to the 10D $N=1$ super Yang-Mills theory. If we split the 10D into 0, (1,2,3,4,5) and (6,7,8,9), then the field $\gs_1$ is the $0^{th}$ component of the 10D gauge field. The 10D gaugino is a Majorana-Weyl spinor, which has 10D chirality +. The 10 D chirality operator can be written as the product of the 6D and 4D chirality operators, those components with plus chirality in both the first six and last four space-time directions become the 5D gaugini, while those with both minus chiralities give rise to the fermion in the hyper-multiplet.
The rotations in 6-7, 8-9 serve as the 5D $N=1$ R-symmetry since it does not mix up the 5D vector- and hyper-multiplets. But the above $R$-charge assignment corresponds to a rotation that would mix the $0^{th}$ and $6^{th}-9^{th}$ space-time directions. Since such a rotation does not commute with the 5D $N=1$ susy, their commutator will necessarily give us a new susy.

But things are trickier than this. To apply localisation, one of course needs a complex that closes off-shell. The 5D $N=1$ vector and hyper-complex (\ref{susy_vect_twist}), (\ref{susy_hyper_twist}) do close off-shell, but the vector part does not rotate covariantly with respect to $\rho$, and so one cannot apply the above logic immediately. One can instead try to start from the 10D $N=1$ formulation, which does reduce to 5D $N=2$ except that the supersymmetry does not close off-shell. But a partial solution to get off-shell closure is provided by Berkovits \cite{Berkovits:1993hx}, who gave a formulation of the 10D $N=1$ supersymmetric Yang-Mills theory with 9 off-shell super charges, which is sufficient for our purpose. Berkovits added seven auxiliary fields $G_{1,\cdots,7}$ and the 10D $N=1$ susy is modified,
\bea
&& \gd A_m=i\ep\Gc_m\gl~,\nn\\
&&\gd\gl=\frac{1}2\slashed{F}\ep+\sum_{i=1}^7\nu_jG_j~,\nn\\
&&\gd G_i=-i\nu_i\slashed{D}\gl~,\nn
\eea
where $A$ is the gauge field, $\gl$ is the Majorana-Weyl gaugino and $\ep$ and $\nu_i$ are Majorana-Weyl spinors satisfying
\bea \nu_j\Gc_m\nu_k-\gd_{jk}\ep\Gc_m\ep =0= \nu_j\Gc_m\ep~.\nn\eea
If one focuses on $\gd\gl$, and traces the susy rule down to 5D carefully, one realises that (\ref{susy_one}) needs to be modified, and the modifications are in red in the left column of the following table
\bea
\begin{array}{c||c|c}
    \textrm{field} & \gd_1 & \gd_2 \\
    \hline
        A & \Psi_1 & \Psi_2\\
    \gs_1 & -\iota_{\sreeb}\Psi_1 & \iota_{\sreeb}\Psi_2\\
    \gs_2 & \iota_{\sreeb}\Psi_2 & \iota_{\sreeb}\Psi_1\\
        B & -\chi_2 & \chi_1\nn\\
        H & -i{\cal L}_{\sreeb}^A\chi_1- [\gs_1, \chi_1]-{\color{red}\gd_1 [\gs_2, B]} & -i{\cal L}_{\sreeb}^A\chi_2 + [\gs_1, \chi_2]+{\color{red}\gd_2 [\gs_2, B]} \\
  \chi_1  & H+{\color{red} [\gs_2, B]} & -i{\cal L}^A_{\sreeb}B+ [\gs_1,  B] \\
  \chi_2  & i{\cal L}^A_{\sreeb} B+ [\gs_1, B] & H-{\color{red} [\gs_2, B]} \\
  \Psi_1  & -i\iota_{\sreeb}F+id_A \gs_1 & -4i(d_A^{\dag} B)^H-id_A \gs_2+i{\cal F}+\gk[\gs_1,\gs_2] \\
  \Psi_2  & 4i(d_A^{\dag} B)^H-id_A \gs_2-i{\cal F}-\gk[\gs_1,\gs_2] & -i\iota_{\sreeb}F-id_A \gs_1 \\
 {\cal F} & \iota_{\sreeb}d_A\Psi_2-4(d_A^{\dag}\chi_2)^H-i [\gs_1,\Psi_2^H] & -\iota_{\sreeb}d_A\Psi_1+4(d_A^{\dag}\chi_1)^H-i [\gs_1, \Psi_1^H]\\
          & +4i\langle \Psi_1, B \rangle -i [\gs_2, \Psi_1^H] & +4i\langle \Psi_2, B\rangle +i [\gs_2,\Psi_2^H]\end{array},\nn\\
          \label{susy_combo}
\eea
Now to obtain the second susy, one can apply the above strategy and consider the commutator
\bea
\gd_2=[\rho,\gd_1]~,\label{new_susy}
\eea
where $\rho$ is defined in (\ref{def_rho}). The result is given in the second column of table
(\ref{susy_combo}).
Note that $\gd_2$ can also be written as
\bea
\gd_2=e^{\frac{\pi}{2}\rho}\gd_1 e^{-\frac{\pi}{2}\rho}~.\label{new_susy_I}\eea
From the closure property of $\gd_1$ and (\ref{new_susy_I}), it is easy to obtain the closure of $\gd_2$
\bea &&\gd_1^2=-i{\cal L}_{\sreeb}-G_{\gs_1+\iota_{\sreeb}A}~,\nn\\
&&\gd_2^2=-i{\cal L}_{\sreeb}-G_{-\gs_1+\iota_{\sreeb}A}~.\label{closure_all}\eea

To make the R-symmetry explicit, we define
\bea
&& \ep=\frac12(\ep_1+i\ep_2)~,~~~\bar\ep=\frac12(\ep_1-i\ep_2)~,\nn\eea
and also recall that $\Psi=\Psi_1+i\Psi_2$, $\chi=\chi_1+i\chi_2$ and $\gs=\gs_1+i\gs_2$, with their R-charges listed in (\ref{list_R_charge}). Define the total susy transformation as $\gd_{\ep}=\ep_1\gd_1+\ep_2\gd_2$
then the susy rule reads
\bea
 \gd_{\ep}A&=&\bar\ep\Psi+\ep\bar\Psi~,\nn\\
\gd_{\ep}\gs&=&-2\bar\ep\iota_{\sreeb}\bar\Psi~,\nn\\
\gd_{\ep}\bar\gs&=&-2\ep\iota_{\sreeb}\Psi~,\nn\\
\gd_{\ep}B&=&-i(\ep\bar\chi-\bar\ep\chi)~,\nn\\
\gd_{\ep}H&=&-i\L_{\sreeb}^A(\bar\ep\chi+\ep\bar\chi)- \ep [\gs,\chi]-\bar\ep [\bar\gs,\bar\chi]-i\iota_{\sreeb}[B,\bar\ep\Psi-\ep\bar\Psi]~,\label{susyN2-first}\\
\gd_{\ep}{\cal F}&=&\ep\big(i\iota_{\sreeb}d_A\bar\Psi-4i(d_A^{\dag}\bar\chi)^H+4i\bra\bar\Psi,B\ket-[\gs,\Psi^H]\big)\nn\\
&&+\bar\ep(-i\iota_{\sreeb}d_A\Psi+4i(d_A^{\dag}\chi)^H+4i\bra\Psi,B\ket+[\bar\gs, \bar\Psi^H])~,\nn\\
\gd_{\ep}\Psi&=&\ep(2{\cal F}-2i\iota_{\sreeb}F-8(d_A^{\dag}B)^H+\gk [\gs,\bar\gs])+2i\bar\ep d_A\bar\gs~,\nn\\
\gd_{\ep}\bar\Psi&=&\bar\ep(-2{\cal F}-2i\iota_{\sreeb}F+8(d_A^{\dag}B)^H-\gk [\gs, \bar\gs])+2i\ep d_A \gs~,\nn\\
\gd_{\ep}\chi&=&2\ep(H-{\cal L}_{\sreeb}^AB)+2i\bar\ep [\bar\gs, B]~,\nn\\
\gd_{\ep}\bar\chi&=&2\bar\ep(H+{\cal L}_{\sreeb}^AB)-2i\ep [\gs, B]~.\nn\eea
One can now obtain the closure property of $Q=\delta_1 + i\delta_2$ and $\bar Q=\delta_1 - i \delta_2$, for this one needs to use either (\ref{new_susy}) or (\ref{new_susy_I}). For example
\bea
\frac12\{Q,\bar Q\}=\frac12\{\gd_1+i\gd_2,\gd_1-i\gd_2\}=\gd_1^2+\gd_2^2=\gd_1^2+e^{\frac{\pi}{2}\rho}\gd^2_1e^{-\frac{\pi}{2}\rho}=-2i({\cal L}_{\sreeb}-iG_{\iota_{\sreeb}A})~,\nn
\eea
while for $Q^2$
\bea Q^2&=&(\gd_1+i\gd_2)^2=\gd_1^2-\gd_2^2+i\{\gd_1,\gd_2\}=\gd_1^2-e^{\frac{\pi}{2}\rho}\gd_1^2e^{-\frac{\pi}{2}\rho}+i\{\gd_1,[\rho,\gd_1]\}\nn\\
&=&-2G_{\gs_1}+2iG_{\gs_2}=-2G_{\bar\gs}\label{closure_Q}
\eea
and similarly $\bar Q^2=-2 G_{\gs}$, where $G_{\gs}$ stands for the infinitesimal gauge transformation with the parameter
 $\gs$, defined in (\ref{gauge_transformation}).

\subsection{The $Q$-exact Action}

 The original supersymmetric action of the vector- and hyper-multiplet written in \cite{Hosomichi:2012ek} does not respect the $U(1)$ symmetry discussed here, and therefore it will not be usable for any calculation of the $N=2$ complex.
  Instead, our action will be $Q$- and $\bar Q$-exact. The purpose of spelling out the detail is to figure out which fields need to be Wick rotated, so as to correctly recover the Haydys-Witten equation.

We will attempt the following three $Q$ exact terms
\bea
S_1 &=&-\frac14\bar Q\big(\Psi,-{\cal F}-i\iota_{\sreeb}F-4(d_A^{\dag}B)^H-\frac12\gk[\gs,\bar\gs]\big) \nn\\
&&-\frac14Q\big(\bar\Psi,{\cal F}-i\iota_{\sreeb}F+4(d_A^{\dag}B)^H+\frac12\gk[\gs,\bar\gs]\big)~,\nn\\
S_2&=&\frac12\bar Q\big(\chi,H+{\cal L}_{\sreeb}^AB+\tilde F\big)+\frac12Q\big(\bar\chi,H-{\cal L}_{\sreeb}^AB+\tilde F\big)~,\nn\\
S_3&=&-\frac18Q(\Psi,\bar Q\bar\Psi)-\frac18\bar Q(\bar\Psi,Q\Psi)~,\nn\eea
where
\bea \tilde F=2i(F_H^{2+}+\frac14B\times B)~,\nn\eea
but there is clearly some freedom in the choice of $\tilde F$, which might hint at the possibility of deforming the equations. But we leave this for future work.

If we only focus on the bosonic term  we have the following expressions
\bea
&&S_1|_{bos}=||{\cal F}+\frac12\gk[\gs, \bar\gs] || ^2+||\iota_{\sreeb}F||^2-16||(d_A^{\dag}B)^H||^2~,\nn\\
&&S_2|_{bos}=2||H||^2-2||{\cal L}_{\sreeb}^AB||^2+2(H,\tilde F)~,\nn\\
&&S_3|_{bos}=(d_A\bar\gs, d_A \gs )~.\nn\eea
The total action is chosen to be the combination
\bea S=S_1+S_2+S_3\nn\eea
and next we integrate out the auxiliary fields $H$
\bea
S&=&||\iota_{\sreeb}F||^2+2||F_H^{+}+\frac{1}{4}B\times B||^2+ (d_A\bar\gs, d_A\gs)-16||(d_A^{\dag}B)^H ||^2\nn\\
&&-2||{\cal L}_{\sreeb}^AB||^2 + \frac14|| [\gs, \bar{\gs}]||^2 + ||{\cal F}||^2~,\nn\eea
where some mixed terms automatically disappear due to the orthogonality.
Observing that the kinetic terms involving $B$ are of the wrong sign, we Wick rotate $B\to -iB$
\bea
&&S=||\iota_{\sreeb}F||^2+2||F_H^{+}-\frac{1}{4}B\times B||^2\nn\\
&&\hspace{2cm}+(d_A\bar\gs, d_A\gs)
+16||(d_A^{\dag}B)^H||^2+
2||{\cal L}_{\sreeb}^AB||^2 +\frac14 || [\gs, \bar{\gs}]||^2~.\nn\eea
Now that the action is positive definite, then the stationary point corresponds to exactly the set of equations (\ref{HW_eqn})
and in addition $d_A \gs=0$, $[\gs, \bar{\gs}]=0$.

The final $N=2$ complex (\ref{final-N2complex}) recorded in the beginning of the section has undergone the Wick rotation $B\to -iB$, as well as the redefinition $\tilde{\cal F}={\cal F}-4(d_A^{\dag}B)^H$ which brings about some simplification to (\ref{susyN2-first}).

\section{Summary}\label{s-summary}

The present work contains two results: the vanishing theorems for the Haydys-Witten equations  and we showed that upon
 certain constraints on the geometry they will collapse to the contact instanton equations. The second result is that we have
  constructed the off-shell $N=2$ twisted supersymetry transformations (\ref{final-N2complex}) for any K-contact manifold (Sasaki and Sasaki-Einstein manifolds are the special cases of such).

The next important question is what does the present $N=2$ theory actually calculate. At the moment we can only give a rough prognosis
 and a proper analysis is left for future investigation. If we consider the case of product manifold $M_5 = M_4 \times S^1$
  then $N=1$ 5D theory will calculate the roof $A$ genus on the moduli space of instantons on $M_4$, since we are dealing effectively with the $N=1$ quantum mechanics on the moduli space of instantons,  see \cite{Baulieu:1997nj}. The $N=2$ theory on $M_5 = M_4 \times S^1$ should correspond to $N=2$ quantum mechanics on the instanton moduli space and thus the result will just produce
   the Euler number of the moduli space,  for these statements about the supersymmetric quantum mechanics and index theorem, see the work of Alvarez-Gaum\'e \cite{Alvarez-Gaume_83_JOP}, \cite{alvarez-gaume1983}.   In particular, if the moduli space contains discrete points, then the Euler number simply counts the number of those points \emph{without sign}, which is of fundamental importance in the application of Haydys-Witten equations in \cite{Witten_FK}.  From point of view of 6D (2,0) theory  we calculate the partition function on
    $M_4 \times T^2$ and this should give rise to the N=4 Vafa-Witten theory on $M_4$ \cite{Witten_FK}.

  Thus it is plausible that when $M_5$ is not a product manifold, we still
    count the solutions for the Haydys-Witten equations. The explicit calculation around the trivial solution $A=0$ gives a contribution
      equal to
     1 (the contribution of the vector-multiplet cancels exactly that of the hyper-multiplet) and it seems
      the same will be true around any other isolated solution of the Haydys-Witten equations simply because the two complexes are isomorphic modulo some technical issues coming from the ghost sector. This would then prove our statement about the counting of solutions, however we have not checked it explicitly. Also at the moment we are not able to produce a coherent understanding of the partition function for
a generic K-contact $M_5$.

\bigskip
{\bf Acknowledgements} We thank  Edward Witten for the initial suggestion toward the relation between contact instantons and
the  Haydys-Witten equations. We thank Sergey Cherkis, Andriy Haydys and Vasily Pestun for helpful discussions.
 We thank two anonymous referees for valuable comments on the manuscript, especially that on the scaling symmetry (\ref{scaling-sym}) of the Haydys-Witten equation.
The research of J.Q. is supported in part by the Luxembourg FNR grant PDR 2011-2, and by the UL grant GeoAlgPhys 2011-2013, and in part by the Max-Planck Institute for Mathematics.
The research of M.Z. is supported in part by Vetenskapsr\r{a}det under grant $\sharp$ 2011-5079.

\appendix
\section{Basics of contact geometry}\label{A-contact}

In this appendix we collect some definitions and facts of the contact geometry.

A manifold $M$ of dimension $2n+1$  is called {\em contact manifold} if it possesses a 1-form $\gk$ (the contact 1-form) such that
\bea
 \gk\wedge(d\gk)^n\neq0\nn~.
 \eea
 The subbundle $\xi\subset TM$ defined by $\xi=\ker \gk$ is called the transverse or horizontal plane.  The data $(M, \xi)$ is called {\em contact structure} on $M$. For a fixed contact form $\gk$  there exists
    a unique vector field $\reeb$ such that $\iota_{\sreeb}\gk=1$, $\iota_{\sreeb}d\gk=0$ which is called {\em the Reeb vector field} (we used small font $\reeb$ to avoid confusion with the curvatures).
  The Reeb foliation is a foliation whose leaves are the Reeb flow, while one needs to pay attention that $\xi$ is not an integrable distribution.

 On a contact manifold $M$ for a fixed contact form $\gk$ one can always choose a metric $g$ \emph{compatible with contact structure} in the following way.
 On the transverse plane $\xi$, there exists a compatible complex structure $J$ in the sense that $d\gk(J-,-)$ is positive definite on $\xi$ and $d\gk(J-,J-)=d\gk(-,-)$.
  The construction is similar to the symplectic case, since $d\gk$ serves as a symplectic structure on $\xi$, see  \cite{MR2682326} for more details.
    Extending $J$ to act as zero on $\reeb$ we can regard $J$ as an endomorphism of $TM$, and it satisfies $J^2=-1+\reeb\otimes \gk$. Thus one can write down a compatible metric
\bea
g=\frac{1}{2}d\gk(J-,-)+\gk\otimes \gk~.\label{can_metric}
\eea
We will lower or raise indices on $J$ without mentioning, so sometimes $J$ will be regarded as a 2-form, and in fact $d\gk=-2J$.
 There are the following important properties of the compatible metric:
 \bea
   \iota_{\sreeb} (\star \omega_p) = (-1)^p \star (\gk \wedge \omega_p)~,\label{used_later}
 \eea
  where $\omega_p$ is $p$-form and $\star$ is the Hodge star operation with respect to $g$ and
 \bea
  {\rm vol}_g = \frac{(-1)^n}{2^n n!} \gk \wedge (d\gk)^n~,
 \eea
 where ${\rm vol}_g$ is the volume form associated with $g$.
 We will refer to $(M, \gk, \reeb, g)$ as
{\em contact metric structure} (be aware that some authors use a different terminology).
For other equivalent definitions and proofs the reader may consult the book \cite{MR2682326}.

We refer to contact metric structure $(M, \gk, \reeb, g)$ as {\em K-contact structure} iff $\reeb$ is Killing with respect to $g$, i.e. ${\cal L}_{\sreeb} g=0$.

Consider  a contact metric structure $(M, \gk, \reeb, g)$ then we can construct  the metric cone $C(M)=M\times \BB{R}^{>0}$, with the metric
\bea
G=r^2g+dr\otimes dr~,\label{cone_metric}
\eea
where $r$ is the coordinate on the $\BB{R}^{>0}$ factor. The endomorphism $J$ has a natural extension ${\cal J}$ to $C(M)$ defined as
\bea
{\cal J}\reeb=-r\partial_r~,~~~{\cal J}r\partial_r=\reeb~.\label{cone-reebJ}
\eea
 Thus the cone $C(M)$ admits the symplectic structure $d (r^2 \gk)$ with a compatible almost complex structure ${\cal J}$.
We say that  $M$ admits a {\em Sasaki structure} if $C(M)$ is a K\"ahler manifold (i.e., the metric $G$ is a K\"ahler metric).
 The covariant constancy of ${\cal J}$ translates to the conditions on $J$
\bea
\bra
Z,(\nabla_XJ)Y\ket=-\gk(Z)\bra X,Y\ket+\bra Z,X\ket\gk(Y)~,\label{integrable_J}
\eea
where $X,Y,Z\in TM$ and $\bra-,-\ket$ is the paring using the metric.
The relation (\ref{integrable_J}) is useful when one needs to decompose the Riemann tensor. For example the following (where $R_{XY}=[\nabla_X,\nabla_Y]-\nabla_{[X,Y]}$)
\bea
\bra U, R_{XY}V\ket-\bra JU, R_{XY}JV\ket=\bra X,U\ket\bra Y,V\ket-\bra X,JU\ket\bra Y,JV\ket-(X\leftrightarrow Y)\label{Riemann_20}
\eea
restricts the (0,2) and (2,0) components  of the Riemann tensor. We will also need the formula
\bea R_{XY}\star J=\textrm{vol}_g(\bra X,c_1Y\ket-(2n-1)\bra X,JY\ket)\label{Riemann_11}\eea
that can be derived from (\ref{Riemann_20}) and the Biancchi identity. Here
$(c_1)_{mn}=R_{mp}J^p_{~n}$ is automatically horizontal and antisymmetric and (1,1) w.r.t $J$.

The Weyl tensor is defined as
\bea W_{ijkl}=R_{ijkl}-\frac{s}{4n(2n+1)}g\bar\wedge g-\frac{1}{2n-1}\left (Ric-\frac{s}{2n+1}g \right )\bar\wedge g~,\label{Weyl_tensor}\eea
where
 $Ric$ is the Ricci tensor $Ric_{ij}=R_{ki~j}^{~~k}$ and $s$ is the Ricci scalar $s=Ric_i^{~i}$. The symbol $\bar\wedge$ is defined as
\bea (A\bar\wedge B)_{ijkl}=A_{ik}B_{jl}-A_{jk}B_{il}-A_{il}B_{jk}+A_{jl}B_{ik}\nn\eea
for two symmetric tensors $A,B$.

Lastly The manifold $M$ is said to be {\em Sasaki-Einstein} if the cone $C(M)$ is Calabi-Yau.

Now let us specialize to five dimensional metric contact manifolds.
By using $\reeb$ and $\gk$ one can decompose any form into its vertical and horizontal components
\be
 \Omega^p (M_5) &=& \Omega^p_V (M_5) \oplus \Omega^p_H (M_5)~,\\
 \ga&=&\gk\iota_{\sreeb}\ga+\iota_{\sreeb}(\gk\wedge\ga)~.
\ee
The space $\Omega^p_V (M_5)$ is orthogonal to $\Omega^p_H (M_5)$ with respect to the scalar product, which is defined using the
 compatible metric $g$
\be
  (\alpha, \beta) =  \int\limits_{M_5} \alpha \wedge {\star} \beta~.
\ee
For the horizontal component one has also the notion of duality, by using the Hodge star
\bea {\star}_{\sreeb}\ga=\iota_{\sreeb}{\star}\ga=(-1)^{\deg\ga}{\star}(\gk\wedge\ga)~.\nn\eea
Then the space $\Omega_H^2(M)$ can be decomposed into
 \be
  \Omega_H^2(M_5) =  \Omega_H^{2+} (M_5) \oplus  \Omega_H^{2-} (M_5)
 \ee
  with the definition
  \be
 \iota_{\sreeb} {\star} \omega_H^\pm = \pm \omega_H^\pm~.\label{hor_self_dual}
\ee
The spaces $\Omega_H^{2\pm} (M)$ are orthogonal to each other under the scalar product too.

On a K-contact manifold, one can develop also the Hodge theory for basic differential forms associated to the Reeb foliation. The basic forms $\Go_B(M)$ are defined as
\bea
\Go_B(M)=\{\ga\in\Go(M)|\iota_{\sreeb}\ga=0=\L_{\sreeb}\ga\}~.\label{basic_form}\eea
It is easy to check that these forms inherit a differential $d^B$ from the De Rham differential. The basic cohomology group $H_B^{\sbullet}(M)$ is defined as the $d^B$-cohomology.
On a K-contact manifold, one has also a paring of differential forms using $\star$. And the associated adjoint operator to $d^B$ is
\bea
(d^B)^{\dag}\ga=(d^{\dag}\ga)^H,\label{d_B_adjoint}\eea
where we remind the reader that $H$ means projection to the horizontal component. The transverse Hodge theory then says that every element of the basic cohomology has a harmonic representative $d^B\ga=(d^B)^{\dag}\ga=0$. Furthermore there are some useful facts about basic cohomology (see section 7 of the book \cite{BoyerGalicki} and references therein):
they are of finite dimension; $H^0_B=H^0$; $H^1_B=H^1$ and if $H^1=0$, one has the exact sequence involving $H_B^2$
\bea
0\to \BB{R}\stackrel{d\gk}{\to} H_B^2 \to H^2\to 0~,\label{exact_sequence}
\eea
where the first map is the multiplication by $d\gk$.

Finally we give the convention for the gamma matrices. The Clifford algebra reads as $\{\Gc_p,\Gc_q\}=2g_{pq}$ and
\bea
\Gc_{i_1\cdots i_k}=\frac{1}{k!}\Gc_{[i_1}\cdots\Gc_{i_k]}~.\nn
\eea
The indices on the gamma matrices will also be raised or lowered with the metric. For a form $A_{i_1\cdots i_k}$, we will use the following three  notations interchangeably
\bea
 \slashed{A}=A\cdotp\Gc=A_{i_1\cdots i_k}\Gc^{i_1\cdots i_k}~.\nn\eea
For further details of the gamma matrix algebra, we refer the reader to \cite{Kallen:2012va} and \cite{Qiu:2013pta}.

\providecommand{\href}[2]{#2}\begingroup\raggedright\endgroup


\begin{thebibliography}{10}

\bibitem{Pestun:2007rz}
V.~Pestun, ``{Localization of gauge theory on a four-sphere and supersymmetric
  Wilson loops},'' \href{http://dx.doi.org/10.1007/s00220-012-1485-0}{{\em
  Commun.Math.Phys.} {\bfseries 313} (2012) 71--129},
\href{http://arxiv.org/abs/0712.2824}{{\ttfamily arXiv:0712.2824 [hep-th]}}.

\bibitem{Kallen:2012cs}
J.~K\"all\'en and M.~Zabzine, ``{Twisted supersymmetric 5D Yang-Mills theory and
  contact geometry},'' \href{http://dx.doi.org/10.1007/JHEP05(2012)125}{{\em
  JHEP} {\bfseries 1205} (2012) 125},
\href{http://arxiv.org/abs/1202.1956}{{\ttfamily arXiv:1202.1956 [hep-th]}}.

\bibitem{Nekrasov:1996cz}
N.~Nekrasov, ``{Five dimensional gauge theories and relativistic integrable
  systems},'' \href{http://dx.doi.org/10.1016/S0550-3213(98)00436-2}{{\em
  Nucl.Phys.} {\bfseries B531} (1998) 323--344},
\href{http://arxiv.org/abs/hep-th/9609219}{{\ttfamily arXiv:hep-th/9609219
  [hep-th]}}.

\bibitem{Baulieu:1997nj}
L.~Baulieu, A.~Losev, and N.~Nekrasov, ``{Chern-Simons and twisted
  supersymmetry in various dimensions},''
  \href{http://dx.doi.org/10.1016/S0550-3213(98)00096-0}{{\em Nucl.Phys.}
  {\bfseries B522} (1998) 82--104},
\href{http://arxiv.org/abs/hep-th/9707174}{{\ttfamily arXiv:hep-th/9707174
  [hep-th]}}.

\bibitem{Harland:2011zs}
D.~Harland and C.~Nolle, ``{Instantons and Killing spinors},''
  \href{http://dx.doi.org/10.1007/JHEP03(2012)082}{{\em JHEP} {\bfseries 1203}
  (2012) 082},
\href{http://arxiv.org/abs/1109.3552}{{\ttfamily arXiv:1109.3552 [hep-th]}}.

\bibitem{Corrigan:1982th}
  E.~Corrigan, C.~Devchand, D.~B.~Fairlie and J.~Nuyts,
  ``{First Order Equations for Gauge Fields in Spaces of Dimension Greater Than Four},''
  {\em Nucl.Phys.} {\bf B214} (1983) 452.
  
  \bibitem{Fan}
  H.~Fan, 
  ``{Half de Rham complexes and line fields on odd-dimensional manifolds},''
   {\em Trans. Amer. Math. Soc.} 
   {\bfseries 348} (1996), 2947-2982. 
  

\bibitem{Kallen:2012va}
J.~K\"all\'en, J.~Qiu, and M.~Zabzine, ``{The perturbative partition function of
  supersymmetric 5D Yang-Mills theory with matter on the five-sphere},''
  \href{http://dx.doi.org/10.1007/JHEP08(2012)157}{{\em JHEP} {\bfseries 1208}
  (2012) 157},
\href{http://arxiv.org/abs/1206.6008}{{\ttfamily arXiv:1206.6008 [hep-th]}}.

\bibitem{Qiu:2013pta}
J.~Qiu and M.~Zabzine, ``{5D Super Yang-Mills on $Y^{p,q}$ Sasaki-Einstein
  manifolds},''
\href{http://arxiv.org/abs/1307.3149}{{\ttfamily arXiv:1307.3149}}.

\bibitem{Qiu:2013aga}
J.~Qiu and M.~Zabzine, ``{Factorization of 5D super Yang-Mills on $Y^{p,q}$
  spaces},'' \href{http://dx.doi.org/10.1103/PhysRevD.89.065040}{{\em
  Phys.Rev.} {\bfseries D89} (2014) 065040},
\href{http://arxiv.org/abs/1312.3475}{{\ttfamily arXiv:1312.3475 [hep-th]}}.

\bibitem{Qiu:2014oqa}
J.~Qiu, L.~Tizzano, J.~Winding, and M.~Zabzine, ``{Gluing Nekrasov partition
  functions},''
\href{http://arxiv.org/abs/1403.2945}{{\ttfamily arXiv:1403.2945 [hep-th]}}.

\bibitem{Kim:2012ava}
H.-C. Kim and S.~Kim, ``{M5-branes from gauge theories on the 5-sphere},''
  \href{http://dx.doi.org/10.1007/JHEP05(2013)144}{{\em JHEP} {\bfseries 1305}
  (2013) 144},
\href{http://arxiv.org/abs/1206.6339}{{\ttfamily arXiv:1206.6339 [hep-th]}}.

\bibitem{Lockhart:2012vp}
G.~Lockhart and C.~Vafa, ``{Superconformal Partition Functions and
  Non-perturbative Topological Strings},''
\href{http://arxiv.org/abs/1210.5909}{{\ttfamily arXiv:1210.5909 [hep-th]}}.

\bibitem{Kim:2012qf}
H.-C. Kim, J.~Kim, and S.~Kim, ``{Instantons on the 5-sphere and M5-branes},''
\href{http://arxiv.org/abs/1211.0144}{{\ttfamily arXiv:1211.0144 [hep-th]}}.

\bibitem{Wolf:2012gz}
M.~Wolf, ``{Contact Manifolds, Contact Instantons, and Twistor Geometry},''
  \href{http://dx.doi.org/10.1007/JHEP07(2012)074}{{\em JHEP} {\bfseries 1207}
  (2012) 074},
\href{http://arxiv.org/abs/1203.3423}{{\ttfamily arXiv:1203.3423 [hep-th]}}.

\bibitem{Baraglia:2014gma}
D.~Baraglia and P.~Hekmati, ``{Moduli Spaces of Contact Instantons},''
\href{http://arxiv.org/abs/1401.5140}{{\ttfamily arXiv:1401.5140 [math.DG]}}.

\bibitem{Pan:2014nha}
Y.~Pan, ``{Note on a Cohomological Theory of Contact-Instanton and Invariants
  of Contact Structures},''
\href{http://arxiv.org/abs/1401.5733}{{\ttfamily arXiv:1401.5733 [hep-th]}}.

\bibitem{Pan:2014bwa}
Y.~Pan, ``{5d Higgs Branch Localization, Seiberg-Witten Equations and Contact
  Geometry},''
\href{http://arxiv.org/abs/1406.5236}{{\ttfamily arXiv:1406.5236 [hep-th]}}.

\bibitem{Taubes_SWW}
C.~H. {Taubes}, ``{The Seiberg-Witten equations and the Weinstein
  conjecture},'' {\em ArXiv Mathematics e-prints} (Oct., 2006) ,
  \href{http://arxiv.org/abs/math/0611007}{{\ttfamily math/0611007}}.

\bibitem{Anderson:2012ck}
L.~Anderson, ``{Five-dimensional topologically twisted maximally supersymmetric
  Yang-Mills theory},'' \href{http://dx.doi.org/10.1007/JHEP02(2013)131}{{\em
  JHEP} {\bfseries 1302} (2013) 131},
\href{http://arxiv.org/abs/1212.5019}{{\ttfamily arXiv:1212.5019 [hep-th]}}.

\bibitem{Witten_SW}
E.~Witten, ``Monopoles and four-manifolds,''
  \href{http://dx.doi.org/10.4310/MRL.1994.v1.n6.a13}{{\em Math. Res. Lett.}
  {\bfseries 1} no.~6, (1994) 769--796}.
  \url{http://dx.doi.org/10.4310/MRL.1994.v1.n6.a13}.

\bibitem{Taubes1994}
C.~H. {Taubes} {\em Mathematical Research Letters} no.~1, (1994) 809–822.

\bibitem{Witten_FK}
E.~Witten, ``{Fivebranes and Knots},''
\href{http://arxiv.org/abs/1101.3216}{{\ttfamily arXiv:1101.3216 [hep-th]}}.

\bibitem{Haydys}
A.~{Haydys}, ``{Fukaya-Seidel category and gauge theory},'' {\em ArXiv
  e-prints} (Oct., 2010) , \href{http://arxiv.org/abs/1010.2353}{{\ttfamily
  arXiv:1010.2353 [math.SG]}}.

\bibitem{Cherkis:2014xua}
S.~A. Cherkis, ``{Octonions, Monopoles, and Knots},''
\href{http://arxiv.org/abs/1403.6836}{{\ttfamily arXiv:1403.6836 [hep-th]}}.

\bibitem{Hosomichi:2012ek}
K.~Hosomichi, R.-K. Seong, and S.~Terashima, ``{Supersymmetric Gauge Theories
  on the Five-Sphere},''
  \href{http://dx.doi.org/10.1016/j.nuclphysb.2012.08.007}{{\em Nucl.Phys.}
  {\bfseries B865} (2012) 376--396},
\href{http://arxiv.org/abs/1203.0371}{{\ttfamily arXiv:1203.0371 [hep-th]}}.

\bibitem{Pan:2013uoa}
Y.~Pan, ``{Rigid Supersymmetry on 5-dimensional Riemannian Manifolds and
  Contact Geometry},''
\href{http://arxiv.org/abs/1308.1567}{{\ttfamily arXiv:1308.1567 [hep-th]}}.

\bibitem{Vafa:1994tf}
C.~Vafa and E.~Witten, ``{A Strong coupling test of S duality},''
  \href{http://dx.doi.org/10.1016/0550-3213(94)90097-3}{{\em Nucl.Phys.}
  {\bfseries B431} (1994) 3--77},
\href{http://arxiv.org/abs/hep-th/9408074}{{\ttfamily arXiv:hep-th/9408074
  [hep-th]}}.

\bibitem{Martelli:2005tp}
D.~Martelli, J.~Sparks, and S.-T. Yau, ``{The Geometric dual of a-maximisation
  for Toric Sasaki-Einstein manifolds},''
  \href{http://dx.doi.org/10.1007/s00220-006-0087-0}{{\em Commun.Math.Phys.}
  {\bfseries 268} (2006) 39--65},
\href{http://arxiv.org/abs/hep-th/0503183}{{\ttfamily arXiv:hep-th/0503183
  [hep-th]}}.

\bibitem{Schmude:2014lfa}
J.~Schmude, ``{Localisation on Sasaki-Einstein manifolds from holomophic
  functions on the cone},''
\href{http://arxiv.org/abs/1401.3266}{{\ttfamily arXiv:1401.3266 [hep-th]}}.

\bibitem{Berkovits:1993hx}
N.~Berkovits, ``{A Ten-dimensional superYang-Mills action with off-shell
  supersymmetry},'' \href{http://dx.doi.org/10.1016/0370-2693(93)91791-K}{{\em
  Phys.Lett.} {\bfseries B318} (1993) 104--106},
\href{http://arxiv.org/abs/hep-th/9308128}{{\ttfamily arXiv:hep-th/9308128
  [hep-th]}}.

\bibitem{Alvarez-Gaume_83_JOP}
L.~Alvarez-Gaume, ``A note on the atiyah-singer index theorem,'' {\em Journal
  of Physics A: Mathematical and General} {\bfseries 16} no.~18, (1983) 4177.
  \url{http://stacks.iop.org/0305-4470/16/i=18/a=018}.

\bibitem{alvarez-gaume1983}
L.~Alvarez-Gaume, ``Supersymmetry and the atiyah-singer index theorem,'' {\em
  Communications in Mathematical Physics} {\bfseries 90} no.~2, (1983)
  161--173. \url{http://projecteuclid.org/euclid.cmp/1103940278}.

\bibitem{MR2682326}
D.~E. Blair, \href{http://dx.doi.org/10.1007/978-0-8176-4959-3}{{\em Riemannian
  geometry of contact and symplectic manifolds}}, vol.~203 of {\em Progress in
  Mathematics}.
\newblock Birkh\"auser Boston, Inc., Boston, MA, second~ed., 2010.
\newblock \url{http://dx.doi.org/10.1007/978-0-8176-4959-3}.

\bibitem{BoyerGalicki}
C.~P. Boyer and K.~Galicki, {\em Sasakian geometry}.
\newblock Oxford Mathematical Monographs. Oxford University Press, Oxford,
  2008.

\end{thebibliography}
\end{document}